\documentclass[letterpaper, 10pt, conference]{ieeeconf}  %

\IEEEoverridecommandlockouts                              %
\usepackage{amsmath} %
\usepackage{amssymb}  %

\usepackage{amsthm}
\usepackage{algorithm}
\usepackage[noend]{algpseudocode}
\usepackage{caption}
\usepackage{subcaption}
\let\labelindent\relax
\usepackage{enumitem}

\usepackage{tikz}
\usepackage{pgfplots}
\usetikzlibrary{external}
\DeclareMathOperator{\tr}{tr}

\newtheorem{theorem}{Theorem}

\newtheorem{lemma}[theorem]{Lemma}
\newtheorem{corollary}[theorem]{Corollary}
\theoremstyle{definition}

\newtheorem{example}[theorem]{Example}

\usetikzlibrary{intersections,fillbetween,backgrounds}

\title{\LARGE \bf Ensuring the Safety of Uncertified Linear State-Feedback Controllers via Switching}

\author{Yiwen Lu and Yilin Mo%
\thanks{The authors are with the Department of Automation and BNRist, Tsinghua University, Beijing, P.R.China. Emails: {\tt\small luyw20@mails.tsinghua.edu.cn, ylmo@tsinghua.edu.cn}.}
\thanks{This work is supported by the National Key Research and Development Program of China under Grant 2018AAA0101601.}
}

\begin{document}

\maketitle
\thispagestyle{empty}
\pagestyle{empty}

\begin{abstract}
   Sustained research efforts have been devoted to learning optimal controllers for linear stochastic dynamical systems with unknown parameters, but due to the corruption of noise, learned controllers are usually uncertified in the sense that they may destabilize the system.
   To address this potential instability, we propose a ``plug-and-play'' modification to the uncertified controller which falls back to a known stabilizing controller when the norm of the state exceeds a certain threshold.
   We show that the switching strategy enhances the safety of the uncertified controller by making the linear-quadratic cost bounded even if the underlying linear feedback gain is destabilizing.
   We also prove the near-optimality of the proposed strategy by quantifying the maximum performance loss caused by switching as asymptotically negligible. Finally, we demonstrate the effectiveness of the proposed switching strategy using simulation on an industrial process example. 
\end{abstract}

\section{INTRODUCTION}

Learning a controller from noisy data for a system with unknown dynamics has been a central topic to adaptive control and reinforcement learning~\cite{aastrom2013adaptive,bertsekas2019reinforcement,recht2019tour,de2021low}.
A main challenge to applying learned controllers is that they are usually \emph{uncertified}, since it can be very difficult to guarantee the stability of such controllers without exact knowledge of the plant. To address this challenge, an alternative to certifying a learned controller is to enhance it with an additional layer of safeguard which prevents the closed-loop system from being destabilized. In particular, assuming the existence of a known stabilizing controller, empirically the safeguard may be achieved by treating the uncertified controller as the primary controller, the stabilizing controller as the fallback, and switching to the fallback when potential unsafety is detected.

Motivated by the above intuition, this paper focuses on the formal design and analysis of a switching strategy with safety guarantee, in the discrete-time Linear-Quadratic Regulation (LQR) setting with Gaussian process noise.
We assume the existence of a known stabilizing linear feedback control law $u = K_0x$, which can be achieved either when the system is known to be open-loop stable (in which case $K_0 = 0$), or through adaptive stabilization methods~\cite{byrnes1984adaptive,faradonbeh2018finite}.
Given an uncertified linear feedback control gain $K_1$,
a modification to the control law $u = K_1 x$ is proposed:
the controller normally applies $u = K_1 x$, but falls back to $u = K_0 x$ for $t$ consecutive steps once $\| x \|$ exceeds a threshold $M$. The proposed strategy is analyzed from both stability and optimality aspects.
In particular, the main results include:
\begin{enumerate}
    \item We prove the LQ cost of the proposed controller is always \emph{bounded}, even if $K_1$ is destabilizing. This fact implies that the proposed strategy enhances the safety of the uncertified controller by preventing the system from being catastrophically destabilized.
    \item Provided $K_1$ is stabilizing, and $M,t$ are chosen properly, we compare the LQ cost of the proposed strategy with that of the linear feedback control law $u = K_1 x$, and quantify the maximum increase in LQ cost caused by switching w.r.t. the strategy hyper-parameters $M, t$ as merely $O(t^{1/4}\exp(-\mathrm{constant}\cdot M^2))$, which decays \emph{super-exponentially} as the switching threshold $M$ tends to infinity.
    \item Inspired by the above quantification, we schedule both $M$ and $t$ to grow logarithmically when the uncertified gain $K_1$ is time-varying and converges to the optimal gain $K^*$ (as is often the case in learning-based schemes), and show that the increase in LQ cost caused by switching converges to zero at a \emph{super-polynomial} rate.
\end{enumerate}
The performance of the proposed switching scheme along with the hyper-parameter schedule is characterized theoretically as well as validated by simulation on an industrial process example. We envision that the switching framework could be potentially applicable in a wider range of learning-based control settings, since it may combine the good empirical performance of learned policies and the stability guarantees of classical controllers, and the ``plug-and-play'' nature of the switching logic may minimize the required modifications to existing learning schemes.

\subsection*{Related Works}

Our proposed strategy belongs to switching control methods.
However, the present paper focuses on applying switching to linear plants, which exhibits desirable properties that may address current challenges in adaptive LQR.
To our knowledge, such application of nonlinear controllers to linear plants has not been widely studied. In the following, works in related areas are reviewed respectively.

\paragraph*{Switched control systems}
Supervisory algorithms have been developed to stabilize switched linear systems~\cite{cheng2005stabilization,sun2005analysis,zhang2010asynchronously}, and other nonlinear systems that are difficult to stabilize globally with a single controller~\cite{prieur2001uniting,el2005output,battistelli2012supervisory}. The (near-)optimality of the controllers, however, are less discussed in the aforementioned works.
Building upon this vein of literature, the idea of switching between certified and uncertified controllers to improve performance was proposed in~\cite{wintz22global}, whose scheme guarantees global stability for general nonlinear systems under mild assumptions, but as a cost of this generality, no quantitative analysis of the performance under switching is provided. From an orthogonal perspective, the present paper specializes in linear systems and proves that switching may induce only negligible performance loss while ensuring safety. 

\paragraph*{Adaptive LQR} Adaptive and learned LQR has drawn significant research attention in recent years, for which high-probability estimation error and regret bounds have been proved for methods including optimism-in-face-of-uncertainty~\cite{abbasi2011regret,cohen2019learning}, thompson sampling~\cite{abeille2018improved}, policy gradient~\cite{fazel2018global}, robust control based on coarse identification~\cite{dean2018regret} and certainty equivalence~\cite{mania2019certainty,faradonbeh2020optimism,faradonbeh2020adaptive,simchowitz2020naive}.
All the above approaches, however, involve applying a linear controller learned from finite noise-corrupted data, which have a nonzero probability of being destabilizing, and may account for the failure probability that occur in the bounds in the aforementioned works. In~\cite{wang2021exact}, a ``cutoff'' method similar to the switching strategy described in the present paper is applied in an attempt to establish almost sure guarantees for adaptive LQR, which are nevertheless asymptotic in nature, and the extra cost caused by switching at a finite time step is not analyzed.
By contrast, the present paper provides a self-contained analysis of the switching strategy, which includes both non-asymptotic and asymptotic results, and is agnostic of learning schemes.

\paragraph*{Nonlinear controller for LQR} Nonlinearity in the control of linear systems has been studied mainly due to practical concerns such as saturating actuators. The performance of LQR under saturation nonlinearity has been studied in~\cite{gokcek2000slqr,gokcek2001lqr,ossareh2016lqr}, which are all based on stochastic linearization, a heuristics that replaces nonlinearity with approximately equivalent gain and bias. By contrast, the present paper treats nonlinearity as a design choice rather than a physical constraint, and provides rigorous performance bounds without resorting to any heuristics. 

\subsection*{Outline}

The remainder of this paper is organized as follows: Section~\ref{sec:problem} introduces the problem setting, describes the proposed switching strategy, and motivates the strategy with a toy example. Section~\ref{sec:theory} contains the analysis of the proposed strategy and the main theoretical results. Section~\ref{sec:simulation} validates the performance of the proposed strategy with a full-size industrial process example. Section~\ref{sec:conclusion} gives concluding remarks.

\subsection*{Notations}

The set of nonnegative integers are denoted by $\mathbb{N}$, and the set of positive integers are denoted by $\mathbb{N}^*$. For a square matrix $M$, $\rho(M)$ denotes the spectral radius of $M$. For a real symmetric matrix $M$, $M\succ 0$ denotes that $M$ is positive definite.
$\|v\|$ denotes the 2-norm of a vector $v$ and $\|M\|$ is the induced 2-norm of the matrix $M$, i.e., its largest singular value. For $P\succ 0$, $\| v \|_P=\|P^{1/2}v\|$ is the $P$-norm of a vector $v$. For two positive semidefinite matrices $P\succ 0, Q\succ 0$, $\|Q \|_P = \| P^{-1/2}QP^{-1/2}\| = \sup_{\| v\|_P = 1} \| v \|_Q^2$. For a random vector $X$, $X\sim N(\mu,\Sigma)$ denotes $X$ is Gaussian distributed with mean $\mu$ and covariance $\Sigma$. $\mathbb{P}(\cdot)$ denotes the probability operator, $\mathbb{E}(\cdot)$ denotes the expectation operator, and $\mathbf{1}_{A}$ is the indicator function of the random event $A$. For functions $f(x), g(x)$ with non-negative values, $f(x) = O(g(x))$ means $\limsup_{x\to\infty} f(x) / g(x) < \infty$, and $f(x) = \Theta(g(x))$ means $f(x) = O(g(x))$ and $g(x)=O(f(x))$.

\section{PROBLEM FORMULATION AND PROPOSED SWITCHING STRATEGY}\label{sec:problem}

Consider the following discrete-time linear plant:
\begin{equation}
    x_{k+1} = Ax_k + Bu_k + w_k,
    \label{eq:dynamics}
\end{equation}
where $k\in \mathbb{N}$ is the time index, $x_k \in \mathbb{R}^n$ is the state vector, $u_k \in \mathbb{R}^m$ is the input vector, and $w_k \in \mathbb{R}^n$ is the process noise. We assume that $x_0 = 0$, and that $w_k \sim N(0,W)$ i.i.d., where $W\succ 0$. We also assume w.l.o.g. that $(A,B)$ is controllable.

We measure the performance of a controller $u_k = u(x_k)$ in terms of the infinite-horizon quadratic cost defined as:
\begin{equation}
    J=\limsup _{T \rightarrow \infty} \frac{1}{T} \mathbb{E}\left[\sum_{k=0}^{T-1} x_{k}^{T} Q x_{k}+u_{k}^{T} R u_{k}\right],
    \label{eq:J}
\end{equation}
where $Q \succ 0, R \succ 0$ are fixed weight matrices specified by the system operator. It is well known that the optimal controller is the linear state feedback $u(x) = K^*x$, whose cost is $J^* = \tr(WP^*)$, where $P^* \succ 0$ satisfies the discrete-time algebraic Riccati equation
\begin{equation}
    P^{*}=Q+A^{T} P^{*} A-A^T P^{*} B\left(R+B^T P^{*} B\right)^{-1} B^{T} P^{*} A,
    \label{eq:dare}
\end{equation}
and the feedback gain can be determined as
\begin{equation}
    K^{*}=-\left(R+B^{T} P^* B\right)^{-1} B^{T} P^{*} A.
    \label{eq:K_opt}
\end{equation}

Now assume that the system and input matrices $A,B$ are unknown to the system operator, then she cannot determine the optimal feedback gain $K^*$ from~\eqref{eq:dare} and~\eqref{eq:K_opt}. Instead, we assume that she has the following two feedback gains at hand:
\begin{itemize}
    \item \emph{Primary gain $K_1$}, typically learned from data, which can be close to $K^*$ but does not have stability guarantees;
    \item \emph{Fallback gain $K_0$}, which is typically conservative but always guaranteed to be stabilizing, i.e., $\rho(A+BK_0) < 1$.
\end{itemize}

The high-level goal of the switching strategy is to choose between primary and fallback controllers in pursuit of both safety and performance. In particular, a block diagram of the closed-loop system under the proposed switching strategy is shown in Fig.~\ref{fig:block_diagram}, and the algorithmic description of the switching logic is shown in Algorithm~\ref{alg:main}. In plain words, the proposed switched control strategy is normally applying $u = K_1 x$, while falling back to $u = K_0 x$ for $t$ consecutive steps once $\| x \|$ exceeds a threshold $M$.
Notice that the switching may be caused by either a destabilizing $K_1$ or a large process noise $w_k$, and the probability of the latter event is controlled by $M$. For the rest of this section, we illustrate the proposed switching mechanism, and motivate the design choices, especially the introduction of the dwell time $t$, by means of a toy example.

\begin{figure}[!htbp]
    \centering
    \includegraphics{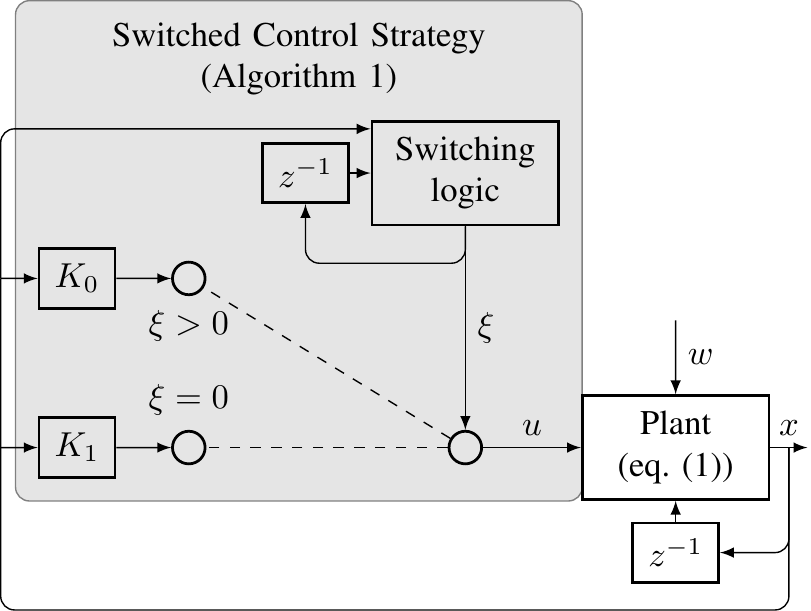}
    \caption{Block diagram of the closed-loop system under the proposed switching strategy. The controller selects $u = K_1 x$ when $\xi = 0$ and $u = K_0 x$ when $\xi > 0$, where $\xi$ is an internal counter determined by the switching logic.}
    \label{fig:block_diagram}
\end{figure}

\renewcommand{\algorithmicrequire}{\textbf{Input:}}
\renewcommand{\algorithmicensure}{\textbf{Output:}}
\begin{algorithm}[!htbp]
    \begin{algorithmic}[1]
        \Require Current state $x$, primary gain $K_1$, fallback gain $K_0$, current counter value $\xi$, switching threshold $M$, dwell time $t$ 
        \Ensure Control input $u$, next counter value $\xi'$
        \If{$\xi > 0$}
            \State $u \gets K_0 x$
        \Else
            \If{$ \| x \|  \geq M$}
                \State $\xi \gets t, u \gets K_0 x$
            \Else
                \State $u \gets K_1 x$
            \EndIf
        \EndIf
        \State $\xi' \gets \max\{\xi - 1, 0 \}$
    \end{algorithmic}
    \caption{Proposed switched control strategy}
    \label{alg:main}
\end{algorithm}

\begin{example}
    Consider a linear system in the form~\eqref{eq:dynamics} with $x \in \mathbb{R}^2, u \in \mathbb{R}$, whose dynamics are described by the matrices
    \begin{equation*}
        A = \begin{bmatrix}
            0.8 & 1 \\ 0 & 0.8
        \end{bmatrix},
        B = \begin{bmatrix}
            0 \\ 1
        \end{bmatrix},
    \end{equation*}
    and the covariance matrix of whose process noise is $W = I$. We choose the LQR weight matrices to be $Q = I, R = 10^{-4}$. We compare the state trajectories under the following learning-based controllers:

    \begin{itemize}
        \item Certainty equivalent controller with no switch, similar to the one in~\cite{simchowitz2020naive,wang2021exact}: $u_k = \hat{K}_k x_k + k^{-1/4}\zeta_k$, where $\hat{K}_k$ is the gain computed from~\eqref{eq:dare},~\eqref{eq:K_opt} by treating $\hat{A}_k, \hat{B}_k$ estimated from past data by least squares as the truth and updated logarithmically often at steps $2^i(i\in \mathbb{N})$, and $k^{-1/4}\zeta_k$ is a random exploration term with $\zeta_k \sim N(0,1)$.
        \item Switched certainty equivalent controllers: the above described certainty equivalent controller modified by the proposed switching strategy, where $\hat{K}_k$ is applied in place of the primary gain $K_1$, and the feedback gain is chosen as $K_0 = 0$ since the system is open-loop stable. The hyper-parameters are chosen to be $M = 10$, and $t = 1, t = 30$ respectively for comparison.
    \end{itemize}

    For a fair comparison between the controllers, the same exploration, parameter learning, and feedback gain updating schemes are applied, and the same realization of process and probing noises $\{w_k \}, \{\zeta_k \}$ are used for benchmarking different controllers.

    One representative trial is visualized in Fig.~\ref{fig:toy_example}, where dashed curve segments stand for the periods of applying the fallback gain.
    It can be observed that the certainty equivalent controller with no switch causes the state to explode exponentially, which can be explained by its application of a destabilizing feedback gain $\hat{K}_k$ ``wrongly'' learned from noisy data. By comparison, the switched controllers trigger the fallback gain $K_0$ (indicated by thick dots in the plot) when the state norm reaches the threshold $M$, and hence keep the state from exploding while waiting for the learned gain $\hat{K}_k$ to improve. The role of the dwell time $t$ is also illustrated: when $t = 1$, i.e., the strategy switches back to the primary gain immediately after the state norm drop below the threshold $M$, oscillation around the threshold may occur. By contrast, if the fallback gain is applied for at least $t$ steps per trigger for a sufficiently large $t$ ($t = 30$ in this example), such oscillation may be avoided. The necessity of the dwell time $t$ is further shown in Subsection~\ref{sec:gap}.

    \begin{figure}[!htbp]
        \centering
        \includegraphics{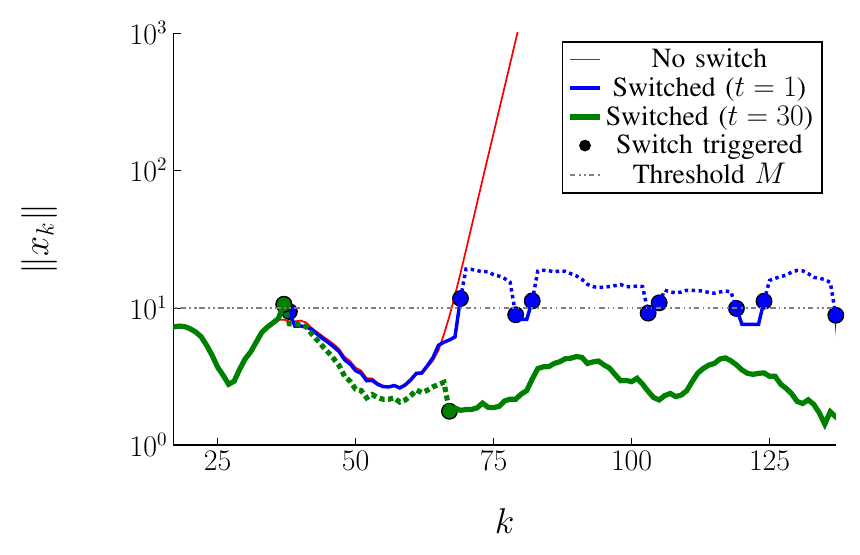}
        \caption{Comparison of state norm trajectories. Dashed curve segments indicate the periods of applying the fallback gain. In the presented trial, the controller with no switch destabilizes the system, the switched controller with $t = 1$ causes the state to oscillate around the switching threshold, and the switched controller with $t = 30$ succeeds to stabilize the system after switching once to the fallback gain.}
        \label{fig:toy_example}
    \end{figure}

    \label{ex:toy_example}
\end{example}

\section{MAIN THEORETICAL RESULTS}\label{sec:theory}

This section is devoted to proving mathematical properties of the proposed switching strategy. Proofs of intermediate results are omitted from the main text to avoid interrupting the flow of the argument.

\subsection{Cost of the controller is bounded with any $K_1$}

In this subsection, we prove that when $M$ and $t$ are fixed, the LQ cost associated with the proposed switched controller is always bounded, regardless of the choice of the underlying primary gain $K_1$ and hyper-parameters $M, t$. This contrasts the fact that the cost associated with the linear controller $u = K_1 x$ is infinity when $K_1$ is destabilizing, and implies that switching enhances the safety of the closed-loop system.

Since the fallback gain $K_0$ is stabilizing, there exists $P_0 \succ 0$ and $0<\rho_0 < 1$ such that
\begin{equation}
    (A+BK_0)^T P_0(A+BK_0) < \rho_0P_0.
    \label{eq:P0}
\end{equation}
The following lemma constructs a quadratic Lyapunov function from $P_0$ and states that the Lyapunov function has bounded expectation:
\begin{lemma}
    Let $V_{0,k} = x_k^T P_0 x_k$, then it holds for any $k$ that
    \begin{equation*}
        \mathbb{E} V_{0,k}<\frac{M^{2} \mathcal{A}^{2} \| P_0 \| + \tr(WP_0)}{1-\rho_0},
    \end{equation*}
    where $P_0 \succ 0$ and $0 < \rho < 1$ are constants that satisfy~\eqref{eq:P0}, and $\mathcal{A} = \max\{ \| A + BK_0 \|, \| A + BK_1 \|\}$.
    \label{lemma:EV}
\end{lemma}

Based on the fact that the LQ cost can be upper bounded by the expectation of the quadratic Lyapunov function $V_{0,k}$ defined above scaled by a constant, we have the following theorem:

\begin{theorem}
    For the closed-loop system under the proposed switching strategy with hyper-parameters $M, t$, the LQ cost defined in~\eqref{eq:J} satisfies
    \begin{equation}
        J<\frac{(M^{2} \mathcal{A}^{2} \| P_0 \| + \tr(WP_0))\| Q_{01}\|_{P_0}}{1-\rho_0},
        \label{eq:J_unstable}
    \end{equation}
    where $\rho_0, P_0$ are constants that satisfy~\eqref{eq:P0}, $\mathcal{A} = \max\{ \| A + BK_0 \|, \| A + BK_1 \|\}$, and $Q_{01} = Q + K_0^T R K_0 + K_1^T R K_1$.
    \label{thm:bounded_cost}
\end{theorem}

\begin{proof}
    According to the switching strategy, we have $u_k \in \{ K_1 x_k, K_0 x_k\}$, which implies
        $x_k^T Q x_k + u_k^T R u_k \leq x_k^T Q_{01} x_k$
    for any $k$, and then by Lemma~\ref{lemma:EV}, $\mathbb{E}(x_k^T Q x_k + u_k^T R u_k)$ is no greater than the RHS of~\eqref{eq:J_unstable} for any $k$.
    The conclusion then follows from the fact
    $J = \limsup _{T \rightarrow \infty} \frac{1}{T} \left.\sum_{k=0}^{T-1} \mathbb{E}( x_{k}^T Q x_{k}+u_{k}^T R u_{k})\right.$.
\end{proof}

\subsection{Upper bound on performance loss caused by switching}
\label{sec:gap}

In this subsection, we quantify the extra LQ cost caused by the conservativeness of switching when $K_1$ is stabilizing. Let $J^{K_1}$ denote the LQ cost associated with the closed-loop system under the linear controller $u = K_1 x$, and $J^{K_1,M,t}$ denote the LQ cost associated with the closed-loop system under our proposed switched controller with primary gain $K_1$ and hyper-parameters $M,t$.
To quantify the behavior of the system under switching, we resort to a common quadratic Lyapunov function for $A+BK_1$ and $(A+BK_0)^t$, which always exists for sufficiently large dwell time $t$. Formally speaking, the following inequalities holds:
\begin{equation}
    \begin{cases}
        (A+BK_1)^T P (A+BK_1) < \rho P,\\
        ((A+BK_0)^t)^T P (A+BK_0)^t < \rho P.
    \end{cases}
    \label{eq:common_lyap}
\end{equation}
Notice that $\rho, P$ that satisfy the first inequality always exist due to the stability of $A+BK_1$, and given $\rho, P$, the quantity $t$ that satisfy the second inequality exists since $\lim_{t\to\infty}(A+BK_0)^t = 0$ by the stability of $A+BK_0$.

The next theorem states moment and tail properties of the closed-loop system under the proposed switching strategy:
\begin{theorem}
    Assume that $\rho_0, P_0$ satisfy~\eqref{eq:P0}, and that $\rho, P, t$ satisfy~\eqref{eq:common_lyap}. Let $\tilde{W} = \sum_{\tau=0}^\infty (A+BK_0)^\tau W((A+BK_0)^\tau)^T$. If the threshold $M \geq M_0 := \sqrt{3 \|\tilde{W} \| \|P \|\|P^{-1} \|} / (1 - \rho^{1/4})$ is large enough, then the following statements hold:
\begin{enumerate}
	\item The fourth moments of $x_k$ is bounded:
    \begin{equation}
        \mathbb{E} \| x_k \|^4_{P_0} \leq 8(\mathcal{Q} \| P_0 \|_P^2+(n^2 + 2n)\| P_0 \|_{\tilde{W}^{-1}}^2),
    \end{equation}
    where 
    \begin{equation*}
        \mathcal{Q}=\frac{6 \rho\left(\operatorname{tr}(\tilde{W} P)^{2}\right)+(1-\rho)\left(n^{2}+2 n\right)\|P\|_{\tilde{W}^{-1}}^{2}}{(1-\rho)\left(1-\rho^{2}\right)},
    \end{equation*}
    \item The probability of not using feedback gain $K_1$ satisfies:
    \begin{equation*}
        \mathbb{P}(u_k \neq K_1 x_k) \leq t \mathcal{E}(M),
    \end{equation*}
    where
    \begin{equation*}
        \mathcal{E}(M) = \frac{2^{n / 2+1}}{\rho^{-1 / 2}-1} \exp \left(-\frac{\left(1-\rho^{1 / 4}\right)^{2} M^{2}}{4 \| \tilde{W} \| \| P \| \| P^{-1} \| }\right),
    \end{equation*}
    which decays {\bf super-exponentially} w.r.t. the threshold $M$.
\end{enumerate}
    \label{thm:moment_prob}
\end{theorem}

We are now ready to state the main theorem of this subsection:

\begin{theorem}
	With $\rho_0, P_0, \rho, P, M_0, \tilde{W}, \mathcal{Q}, \mathcal{E}$ defined the same as in Theorem~\ref{thm:moment_prob}, assuming that the dwell time $t$ satisfies \eqref{eq:common_lyap} and the threshold $M \geq M_0$, it holds that
    \begin{equation}
        J^{K_1,M,t} - J^{K_1} \leq 2 C_1C_2\mathcal{G}+(C_2^2+C_3) \mathcal{G}^2,
        \label{eq:gap}
    \end{equation}
    where
    \begin{align*}
        & \mathcal{G} = C_4 (t\mathcal{E}(M))^{1/4},  \quad C_{1}=\sqrt{\tr(WP) \| Q_1 \|_P/ (1 - \rho)}, \\
        & C_2 = \| \Delta_1 \| \| Q_1 \|_{P_0}\sum_{s=0}^\infty \| A_1^s \|_{Q_1}, \quad C_3 = \| \Delta_2 \| \| P_0^{-1} \|, \\
        & C_4 = 2^{3/4}(\mathcal{Q} \| P_0 \|_P^2+(n^2 + 2n)\| P_0 \|_{\tilde{W}^{-1}}^2)^{1/4}, \\
        & {Q}_1 = Q + K_1^TRK_1, A_1 = A+BK_1, \\ & \Delta_1 = B(K_0 - K_1), \quad \Delta_2 = K_0^T R K_0 - K_1^T R K_1.
    \end{align*}
    \label{thm:gap}
\end{theorem}

\begin{proof}
    Let $\check{x}_0 = x_0$ and $\check x_{k+1}= A_1\check x_k + w_k$, then $J^{K_1} = \lim_{T \to \infty}\sum_{k = 0}^{T-1}\mathbb{E}\| \check{x}_k \|_{Q_1}^2 / T$. On the other hand, we have $J^{K_1,M,t} = \limsup_{T \to \infty}\sum_{k = 0}^{T-1}\mathbb{E}[x_k^T Q x_k + u_k^T R u_k] / T$, and therefore we only need to prove $\mathbb{E}[x_k^T Q x_k + u_k^T R u_k] - \mathbb{E}\| \check{x}_k \|_{Q_1}^2$ is no greater than the RHS of~\eqref{eq:gap} for any $k$. Notice that
    \begin{equation*}
        x_k^T Q x_k + u_k^T R u_k - \|\check{x}_k \|_{Q_1}^2 = \| x_k \|_{Q_1}^2  - \|\check{x}_k \|_{Q_1}^2 + x_k^T \Delta_2 x_k \mathbf{1}_{F_k},
    \end{equation*}
    where $F_{k} = \{ u_k \neq K_1 x_k \}$ denotes the event that the fallback mode is active and the gain $K_1$ is not applied at step $k$. We will next bound $\mathbb{E}\| x_k \|_{Q_1}^2  - \mathbb{E}\|\check{x}_k \|_{Q_1}^2$ and $\mathbb{E}(x_k^T \Delta_2 x_k \mathbf{1}_{F_k})$ respectively.

    \paragraph{Bounding $\mathbb{E}\| x_k \|_{Q_1}^2  - \mathbb{E}\|\check{x}_k \|_{Q_1}^2$} \label{par:bound_norm_diff} Notice that 
    \begin{equation*}
    	x_k = A_1x_{k-1}+w_{k-1}+ \Delta_1 x_{k-1}\mathbf{1}_{F_{k-1}},
    \end{equation*}
    and by recursively applying this expansion, we get
    \begin{align*}
	    x_k & = A_1^{k} x_{0}+\sum_{s=0}^{k-1}A_1^{k-s-1} (w_{s} + \Delta_1 {x}_{s} \mathbf{1}_{F_{s}})\\
        & = \check x_k +\sum_{s=0}^{k-1}A_1^{k-s-1}  \Delta_1 {x}_{s} \mathbf{1}_{F_{s}}.
    \end{align*}
    Hence,
    \begin{equation*}
        \| x_k \|_{Q_1} \leq  \| \check{x}_k \|_{Q_1} + \| \Delta_1 \|_{Q_1} \sum_{s=0}^{k-1}\|A_1^{k-s-1}\|_{Q_1} \| {x}_{s} \|_{Q_1}\mathbf{1}_{F_{s}}.
    \end{equation*}
    From the fact that $\mathbb{E}(\sum_{i=1}^n X_i)^{2} \leq (\sum_{i=1}^n \sqrt{\mathbb{E}X_i^2})^2$ for any random variables $X_1,\ldots,X_n$, we have
    \begin{align*}
        & \mathbb{E}\| x_k \|_{Q_1}^2 \leq \left(\vphantom{\sum_1^2}\sqrt{\mathbb{E}\| \check{x}_k \|_{Q_1}^2}  + \right. \\&  \quad \left.\| \Delta_1 \|_{Q_1} \| Q_1 \|_{P_0} \sum_{s=0}^{k-1}\|A_1^{k-s-1}\| \sqrt{\mathbb{E}[ \| x_{s}\|_{P_0}^2 \mathbf{1}_{F_{s}} ]}\right)^2
    \end{align*}
    where by Cauchy-Schwarz inequality and Theorem~\ref{thm:moment_prob}, we have $\sqrt{\mathbb{E}\left[ \| x_{s}\|_{P_0}^2 \mathbf{1}_{F_{s}} \right]}  \leq (\mathbb{E}\| x_s \|_{P_0}^4)^{1/4}\mathbb{P}(F_s)^{1/4} \leq \mathcal{G}$ for any $s$. Furthermore, we have $\mathbb{E}\| \check{x}_k \|_{Q_1}^2 \leq C_1^2$ guaranteed by~\eqref{eq:common_lyap}. Hence, we have
    \begin{equation*}
        \mathbb{E}\| x_k \|_{Q_1}^2  - \mathbb{E}\|\check{x}_k \|_{Q_1}^2 \leq 2 C_1 C_2 \mathcal{G} + C_2^2 \mathcal{G}^2.
    \end{equation*}
    
    \paragraph{Bounding $\mathbb{E}(x_k^T \Delta_2 x_k \mathbf{1}_{F_k})$} We have $x_k^T \Delta_2 x_k \mathbf{1}_{F_k} \leq \|\Delta_2\|\| x_k\|^2\mathbf{1}_{F_k} \leq C_3 \| x_k\|_{P_0}^2\mathbf{1}_{F_k}$. Following a similar argument to part (a), we get $$\mathbb{E}(x_k^T \Delta_2 x_k \mathbf{1}_{F_k}) = C_3 \mathcal{G}^2.$$

    Combining the above two parts, we obtain the desired conclusion.
\end{proof}

The below corollary following directly from Theorem~\ref{thm:gap}, which indicates that under proper choice of $t$, the performance loss caused by switching can decay super-exponentially $M$ is enlarged.

\begin{corollary}
    When $K_1$ is held constant, and $M, t$ are varied, we have
    \begin{equation}
        J^{K_1,M,t} - J^{K_1} = O(t^{1/4}\exp(-cM^2))
        \label{eq:gap_bigO}
    \end{equation}
    as $M \rightarrow \infty, t \rightarrow \infty, t^{1/4} \exp \left(-c M^{2}\right) \rightarrow 0$, where $c = (1 - \rho^{1/4})^2 / (16 \| \tilde{W} \| \| P \| \| P^{-1} \|)$ is a system-dependent constant.
\end{corollary}

\subsection{Super-polynomial decay of performance loss with varying $K_1,M,t$}

This subsection builds upon Theorem~\ref{thm:gap} to prove shown the performance loss $J^{K_1,M,t} - J^{K_1}$ decays faster than any polynomial when $K_1$ converges to the optimal gain $K^*$ and both $M,t$ grow logarithmically, which provides a practical schedule for choosing $M$ and $t$ without relying on any unknown system-dependent constants.

Let $K^*$ be the optimal gain defined in~\eqref{eq:K_opt}, and let the switching strategy vary with the time index $k$. Assume the primary gain $K_{1,k}$ used at step $k$ satisfy $\lim_{k \to \infty} K_{1,k} = K^*$, as is often the case in learning-based schemes, then it follows that there exists $P \succ 0, 0 < \rho < 1, t\in \mathbb{N}$ such that~\eqref{eq:common_lyap} holds \emph{uniformly} for all sufficiently large $k$. 
When choosing $M_k = \Theta(\log(k)), t_k = \Theta(\log(k))$, the assumptions of Theorem~\ref{thm:gap} the $M,t$ are sufficiently large hold automatically when $k$ is large enough. By invoking Theorem~\ref{thm:gap} and substituting $M_k = \Theta(\log(k)), t_k = \Theta(\log(k))$ into~\eqref{eq:gap_bigO},
it holds
\begin{equation*}
    J^{K_{1,k}, M_k, t_k} - J^{K_{1,k}} = O(k^{-\alpha}),
\end{equation*}
for any $\alpha > 0$. The above described logarithmic schedule for choosing $M,t$ and the corresponding super-polynomial decay rate of the extra cost is validated by a simulation in the next section.
\section{NUMERICAL SIMULATION}\label{sec:simulation}

This section demonstrates the effectiveness and near-optimality of the proposed switching scheme using an industrial process example.

\begin{example}
    Tennessee Eastman Process (TEP) is a commonly used process control system proposed by~\cite{downs1993plant}. In this simulation, we consider a simplified version of TEP similar to the one in~\cite{liu2020online} with full-state-feedback. The system has state dimension $n = 8$ and input dimension $m = 4$. The system is open-loop stable, and therefore the fallback controller is chosen as $K_0 = 0$. The LQ weight matrices and the covariance of process noise are chosen as $Q=I,R=I,W=I$. A time-varying switched control strategy is deployed, whose hyper-parameters vary as $M_k = \log(k+1), t_k = \lfloor \log(k+1) \rfloor$, and whose underlying primary gain $K_{1,k}$ is updated at steps $k = 2^i(i\in \mathbb{N})$ according to $\hat{A}, \hat{B}$ learned by the least-squares algorithm similar to~\cite{wang2021exact}. At each step $k$, we compare the LQ cost $J^{K_{1,k}, M_k, t_k}$ associated with our proposed switched controller against the LQ cost $J^{K_{1,k}}$ associated with the linear controller of the underlying primary gain. Since we do not know of any method to evaluate $J^{K_{1,k}, M_k, t_k}$ analytically, it is approximated by random sampling of $1000$ trajectories, each with length $100$, at each $k$.

    \begin{figure}[!htbp]
        \centering
        \includegraphics{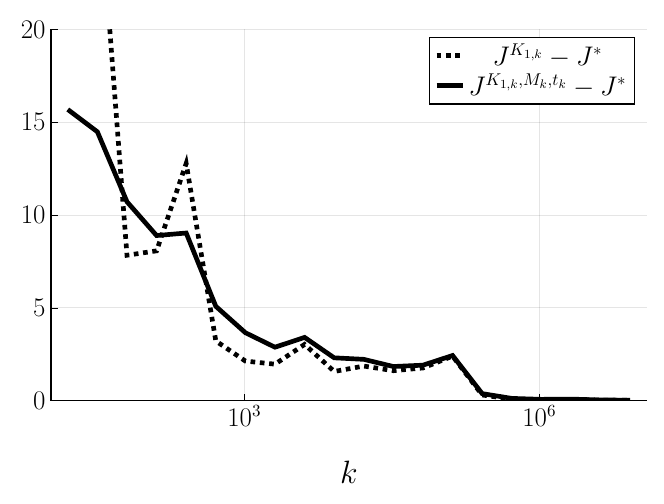}
        \caption{Simulation results on TEP. The proposed switching guarantees finite cost even if the underlying primary gain is destabilizing, and the extra cost incurred converges to zero.}
        \label{fig:compare}
    \end{figure}

    The simulation results are shown in Fig.~\ref{fig:compare}. It can be observed that the switching strategy guarantees finite LQ cost even when the primary gain $K_1$ is destabilizing, as is the case in the initial stage. As $K_1$ gradually improves, the switching strategy begins to incur an extra cost, but the extra cost converges to zero as $k$ increases. The the convergence rate of the extra cost is visualized using a log-log plot in Fig.~\ref{fig:diff}, where super-polynomial convergence can be observed.

\begin{figure}[!htbp]
    \centering
    \includegraphics{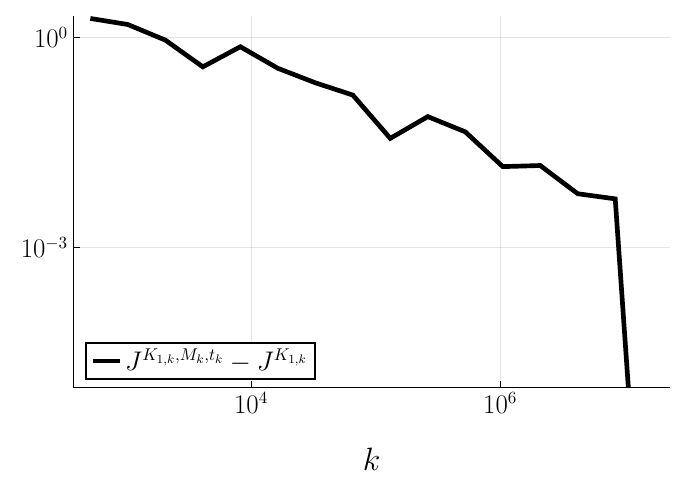}
    \caption{Super-polynomial decay of the extra cost}
    \label{fig:diff}
\end{figure}

\end{example}
\section{CONCLUSION}\label{sec:conclusion}

This paper introduces a plug-and-play switching strategy which enhances the safety of uncertified linear state-feedback controllers, and quantify the extra cost caused by such switching strategy. Practical schedules for varying the strategy hyper-parameters over time are also designed, such that super-polynomial decay of the extra cost can be achieved. Future directions include extending the switching strategy with near-optimality guarantee to more general classes of systems.

\bibliographystyle{IEEEtran}
\bibliography{ref.bib}

\end{document}